\RequirePackage{fix-cm}
\documentclass[smallcondensed]{svjour3}     
\smartqed  
\usepackage{graphicx}
\usepackage[dvipsnames]{xcolor}
\usepackage{amssymb}
\usepackage{float}
\usepackage{amsmath}
\usepackage[title]{appendix}
\newcommand{\newc}{\newcommand}
\newc{\beq}{\begin{equation}}
\newc{\eeq}{\end{equation}}
\newc{\kt}{\rangle}
\newc{\br}{\langle}
\newc{\beqa}{\begin{eqnarray}}
\newc{\eeqa}{\end{eqnarray}}
\newc{\pr}{\prime}
\newc{\longra}{\longrightarrow}
\newc{\ot}{\otimes}
\newc{\rarrow}{\rightarrow}
\newc{\h}{\hat}
\newc{\bom}{\boldmath}
\newc{\btd}{\bigtriangledown}
\newc{\al}{\alpha}
\newc{\be}{\beta}
\newc{\ld}{\lambda}
\newc{\sg}{\sigma}
\newc{\p}{\psi}
\newc{\eps}{\epsilon}
\newc{\om}{\omega}
\newc{\mb}{\mbox}
\newc{\tm}{\times}
\newc{\hu}{\hat{u}}
\newc{\hv}{\hat{v}}
\begin{document}

\title{Quantum Discord and Logarithmic Negativity in the Generalized $n$-qubit Werner State
}


\author{M. S. Ramkarthik         \and
        Devvrat Tiwari  \and
        Pranay Barkataki
}


\institute{M. S. Ramkarthik \at
              \email{msramkarthik@phy.vnit.ac.in}             \\
           \and
            Devvrat Tiwari   \at
           	  \email{devvrat6@gmail.com}\\
           \and
           Pranay Barkataki \at
              \email{pranaybarkataki@students.vnit.ac.in}\\
              \emph{Department of Physics, Visvesvaraya National Institute of Technology, Nagpur - 440010}
}

\date{Received: date / Accepted: date}

\maketitle

\begin{abstract}
Quantum Discord (QD) is a measure of the total quantum non-local correlations of a quantum system. The formalism of quantum discord has been applied to various two-qubit mixed states and it has been reported that there is a non-zero quantum discord even when the states are unentangled. To this end, we have calculated the Quantum Discord for higher than two qubit mixed state, that is, the generalized $n$-qubit Werner state with a bipartite split. We found that the QD saturates to a straight line with unit slope in the thermodynamic limit. Qualitative studies of entanglement between the two subsystems using logarithmic negativity revealed that the entanglement content between them increases non-uniformly with the number of qubits leading to its saturation. We have proved the above claims both analytically and numerically.
\end{abstract}
\section{Introduction}
\label{intro}
Quantum non-local correlations has been one of the most unique and exclusive phenomena of the quantum world, which has no classical analogue. Over the years it was assumed that quantum entanglement \cite{horodecki} is the only source of quantum non-local correlations. However in the recent years, numerous evidences indicated in the direction that quantum entanglement is not the only non-local correlation that exist between two subsystems. Over the years it was shown that there are states that possess non-local behaviour but are not entangled \cite{knill},\cite{Bennett}. Zurek and Ollivier \cite{zurek}, Henderson and Vedral \cite{vedral} independently gave a measure for the quantification of such non-local correlations known as \textit{Quantum Discord}. For a bipartite composite quantum system $\rho_{AB}$, where $A$ and $B$ are the individual subsystems, the QD is defined as follows,
\begin{equation}
\mathcal{D}(A:B) = I(A:B) - J(A|B),
\end{equation}
where $I(A:B)$ and $J(A|B)$ are the total and classical correlations present between the two subsystems, respectively. Both of these terms are defined as follows,  $I(A:B) = S(A) + S(B) - S(A,B)$ and $J(A|B)=S(A) - S(A|B)$, where $S(A)$ and $ S(B)$ are the von Neumann entropies of the subsystem states $\rho_{A}$ and $\rho_B$, respectively. The terms $S(A,B)$ and $S(A|B)$ are the joint von Neumann entropy and the quantum conditional entropy of the system, respectively. Based on these definitions the expression of the QD can be rewritten as,
\begin{equation}
\mathcal{D}(A:B) = S(B) - S(A,B) + S(A|B). \label{discord_final_1}
\end{equation}
The quantum conditional entropy $S(A|B)$ is given as,
\begin{equation}
S(A|B) = \min_{\{\Pi_k\}}\sum_{k=1}^{{\cal H}_B}p_kS(\rho_{A|\Pi_k}), \label{quant_cond_ent}
\end{equation}
where ${\cal H}_B$ is the Hilbert space dimension of the subsystem $B$. In Eq.(\ref{quant_cond_ent}), $\rho_{A|\Pi_k}$ is the post-measurement state for the subsystem $A$ when a measurement is performed on the subsystem $B$, and $p_k=Tr(\Pi_k^\dagger\Pi_k\rho_{AB})$ is the probability associated with the measurement operators $\Pi_k$. The state $\rho_{A|\Pi_k}$ can be explicitly written as,
\begin{eqnarray}
\rho_{A|\Pi_k} &=& \frac{1}{p_k}Tr_B\left(\Pi_k\rho_{AB}\Pi_k\right). \label{state_after_meas}
\end{eqnarray}
In this paper, we are constructing the generalized Werner state of $n$ qubits as given in ref.\cite{siewert},\cite{cirac}. When we reduce this state to two qubits, it becomes the famous Werner state \cite{werner}. We study the QD and the logarithmic negativity between any single qubit ($B$ subsystem) and the remaining $n-1$ qubits ($A$ subsystem) of the generalized $n$-qubit Werner state. It is to be noted that we can calculate the QD by performing a measurement over any one of the subsystems $A$ or $B$. The subsystem $B$ may contain any arbitrary number of qubits between $1$ to $n-1$, and accordingly the subsystem size of $A$ will vary. However, here we are performing a measurement over subsystem $B$ and limiting it to one-qubit, because the minimization in Eq.(\ref{quant_cond_ent}) depends upon $2^m$ parameters of the measurement operators $\Pi_k$, where $m$ ($1\leq m <n$) is the number of qubits in the subsystem $B$. Therefore, for $m>1$ it is difficult to track the solution both analytically and numerically as it is non-trivial and there exist no generic prescription for the same. For $m=1$, the general measurement operators are,
\begin{align}
\Pi_1 = \mathbb{I}_A \otimes |u\kt_B {}_B\br u| && \text{and} && \Pi_2 = \mathbb{I}_A \otimes |v\kt_B {}_B\br v|, \label{measure_op}
\end{align}
where $|u\kt = \cos(\theta)|0\kt + e^{\textit{i}\phi}\sin(\theta)|1\kt$ and $|v\kt = \sin(\theta)|0\kt - e^{\textit{i}\phi}\cos(\theta)|1\kt$, and two parameters, $\theta$ and $\phi$, vary in the range, $0\leq\theta\leq\pi/2$ and $0\leq\phi\leq 2\pi$. In section \ref{sec_QD}, we analytically calculate QD between the subsystems $A$ and $B$, and in section \ref{N_QD}, we derive the linear relationship between QD and mixing probability $p$ at the thermodynamic limit ($n$ is very large, that is $n\rightarrow \infty$). The logarithmic negativity between the subsystems $A$ and $B$ has been investigated in section \ref{negv}.
\section{Calculation of the QD for the generalized $n$-qubit Werner state}
\label{sec_QD}
The QD for the two-qubit Werner state has already been studied in ref.\cite{zurek}, however in this paper we study the variation of the non-local correlations with respect to the total number of qubits $n$. The generalized $n$-qubit Werner state is given as,
\begin{equation}
\rho_{W_{AB}} = p|\phi\kt_{AB}\hspace{0.03cm}{}_{AB}\br \phi| + \frac{(1-p)}{2^n}\mathbb{I}_{AB} , \label{n_qubit_werner1}
\end{equation}
where $|\phi\rangle_{AB}$ is the $n$-qubit GHZ state \cite{ghz} given as, $|\phi\rangle_{AB} = \frac{1}{\sqrt{2}} \left(\right.|0\kt^{\otimes n-1}_{A}|0\rangle_B + |1\kt^{\otimes n-1}_{A}$ $|1\rangle_B \left.\right)$, and $p$ is the mixing probability such that $0\leq p\leq 1$. In Eq.(\ref{n_qubit_werner1}), $\mathbb{I}_{AB}/2^n$ is the $n$-qubit maximally mixed state. For calculating the QD between subsystems $A$ and $B$ of the state $\rho_{W_{AB}}$, the first step is the calculation of the von Neumann entropies in Eq.(\ref{discord_final_1}). To this end, we calculate the reduced density matrix $\rho_B$ as shown below,
\begin{eqnarray}
\rho_B &=& \left\{2^{n-1}\left(\frac{1-p}{2^n}\right) + \frac{p}{2}\right\}|0\kt\br 0| + \left\{2^{n-1}\left(\frac{1-p}{2^n}\right) + \frac{p}{2}\right\}|1\kt\br 1|, \\
&=& \frac{1}{2}\left\{|0\kt\br 0| + |1 \kt\br 1|  \right\}.
\end{eqnarray}
The eigenvalues of $\rho_B$ are $\frac{1}{2}$ and $\frac{1}{2}$. Therefore, the von Neumann entropy for the subsystem $B$ can be calculated as,
\begin{eqnarray}
S(B) &=& -\frac{1}{2}\log_2 \left(\frac{1}{2}\right) - \frac{1}{2}\log_2 \left(\frac{1}{2}\right)= 1. \label{n_qubit_sb}
\end{eqnarray}
Let us define $2^n = N$, for compactness of notation. The matrix representation of the state $\rho_{W_{AB}}$ is given below,
\begin{equation}
\rho_{W_{AB}} = \begin{bmatrix}
a_{11}&0&0&\dots&a_{1 N}\\
0&a_{22}&0&\dots& 0\\
0&0&a_{33}&\dots& 0\\
\vdots& & &\ddots&\vdots\\
a_{N 1}& & & & a_{N N}
\end{bmatrix}_{N \times N}. \label{n_qubit_werner_matrix}
\end{equation}
Using Eq.(\ref{n_qubit_werner1}), the matrix elements $a_{ij}$'s in Eq.(\ref{n_qubit_werner_matrix}) can be written as,
\begin{equation}
a_{11} =a_{N N} =  \left(\frac{1-p}{2^n} + \frac{p}{2}\right);\;\; a_{1 N} = a_{N 1} = \frac{p}{2}, \label{n_qubit_rest_factors1}
\end{equation}
and rest of the matrix elements are,
\begin{align}
a_{22} = a_{33} = a_{44} = \ldots = a_{N-1\;N-1} = \left(\frac{1-p}{2^n}\right). \label{n_qubit_direct_factors1}
\end{align}
The characteristic equation of the matrix written in Eq.(\ref{n_qubit_werner_matrix}) can be nicely factored out as,
\begin{eqnarray}
(a_{22}-\ld)(a_{33} - \ld)\cdots(a_{N-1\;N-1}-\ld)(\ld^2 - \ld(a_{11}+ a_{NN}) + a_{11}a_{NN} \nonumber \\ - a_{N1}a_{1N}) = 0. \label{n_qubit_character}
\end{eqnarray}
The roots (eigenvalues) of the above characteristic equation are given below,
\begin{equation}
\ld_1 = \ld_2 = \ld_3= \ldots =\ld_{N-2} = \left(\frac{1-p}{2^n}\right).\label{eigv_N2}
\end{equation}
The remaining two eigenvalues $\ld_{N-1}$ and $\ld_N$ can be written as,
\begin{eqnarray}
\ld_{N-1} &=& \frac{1}{2}\left\{(a_{11}+a_{N N}) + \sqrt{(a_{11}-a_{N N})^2 + 4 a_{1 N}a_{N 1}}\right\}\nonumber \\
&=&\left(\frac{1 + (2^n - 1)p}{2^n}\right);\label{eigv_N1} \\
\ld_{N} &=& \frac{1}{2}\left\{(a_{11}+a_{N N}) - \sqrt{(a_{11}-a_{N N})^2 + 4 a_{1 N}a_{N 1}}\right\}\nonumber \\
&=&\left(\frac{1-p}{2^n}\right).\label{eigv_N}
\end{eqnarray}%
Based on the eigenvalues calculated from Eqs.(\ref{eigv_N2}-\ref{eigv_N}), we can directly write the joint entropy $S(A,B)$ as follows,
\begin{eqnarray}
S(A,B) &=&  (N-2) \left\{-\left(\frac{1-p}{2^n}\right)\log_2 \left(\frac{1-p}{2^n}\right)\right\} -\left(\frac{1+(2^n -1)p}{2^n}\right) \nonumber \\ && \log_2 \left(\frac{1+(2^n -1)p}{2^n}\right) -  \left(\frac{1-p}{2^n}\right)\log_2 \left(\frac{1-p}{2^n}\right),\\
S(A,B) &=&- (2^n - 1) \left(\frac{1-p}{2^n}\right)\log_2 \left(\frac{1-p}{2^n}\right) -\left(\frac{1+(2^n -1)p}{2^n}\right) \nonumber \\ && \log_2 \left(\frac{1+(2^n -1)p}{2^n}\right). \label{n_qubit_sab}
\end{eqnarray}
The next step is to find the quantum conditional entropy $S(A|B)$ of the system by performing a measurement over the subsystem $B$ by the set of operators defined in Eq.(\ref{measure_op}). To this end, we operate the measurement operator $\Pi_1$ on the state $\rho_{W_{AB}}$, as shown below,
\begin{eqnarray}
\Pi_1 \rho_{W_{AB}}\Pi_1 &=& \frac{p}{2}\cos^2(\theta)|0\kt^{\otimes n-1}_A |u\kt_B{}_A\br 0|^{\otimes n-1} {}_B\br u|  + \frac{p}{2}e^{\textit{i}\phi}\cos(\theta)\sin(\theta)\nonumber \\ &&|0\kt^{\otimes n-1}_A|u\kt_B {}_A\br 1|^{\otimes n-1}{}_B\br u| + \frac{p}{2}e^{-\textit{i}\phi}\cos(\theta)\sin(\theta)|1\kt^{\otimes n-1}_A|u\kt_B \nonumber \\ && {}_A\br 0|^{\otimes n-1}{}_B\br u|  + \frac{p}{2}\sin^2(\theta)|1\kt_A^{\otimes n-1}|u\kt_B{}_A\br 1|^{\otimes n-1}{}_B\br u| \nonumber \\ && + \sum_{i=0}^{2^{n-1}-1}\left(\frac{1-p}{2^n}\right)|i\kt_A |u\kt_{B}{}_{A}\br i| {}_{B}\br u|\{\cos^2(\theta) + \sin^2(\theta)\}.
\end{eqnarray}
The associated probability for the outcome of the measurement operator $\Pi_1$ is $p_1$, and it can be calculated as,
\begin{eqnarray}
p_1 &=& tr(\Pi_1\rho_{W_{AB}}\Pi_1),\\
&=& \frac{p}{2}\cos^2(\theta) + \frac{p}{2}\sin^2(\theta) + 2^{n-1}\left(\frac{1-p}{2^n}\right)=  \frac{p}{2} + \left(\frac{1-p}{2}\right) = \frac{1}{2}. \label{n_qubit_p1}
\end{eqnarray}
The post measurement state for the subsystem $A$ is written as,
\begin{eqnarray}
\rho_{A|\Pi_1} &=& \frac{1}{p_1}tr_B(\Pi_1\rho_{W_{AB}}\Pi_1),\\
&=& 2 \times \frac{p}{2} \cos^2(\theta)|0\kt^{\otimes n-1}\br 0|^{\otimes n-1}  + 2 \times \frac{p}{2}e^{\textit{i}\phi}\cos(\theta)\sin(\theta)|0\kt^{\otimes n-1}\br 1|^{\otimes n-1} \nonumber \\ && + 2 \times \frac{p}{2}e^{-\textit{i}\phi}\cos(\theta)\sin(\theta)|1\kt^{\otimes n-1}\br 0|^{\otimes n-1} + 2 \times \frac{p}{2} \sin^2(\theta)|1\kt^{\otimes n-1}\br 1|^{\otimes n-1} \nonumber \\ && + 2\sum_{i =0}^{2^{n-1}-1}\left(\frac{1-p}{2^n}\right)|i\kt\br i|,\label{rdm_A}
\end{eqnarray}
The reduced density matrix of the subsystem $A$ in Eq.(\ref{rdm_A}) can be further simplified to be,
\begin{eqnarray}
\rho_{A|\Pi_1} &=&  p \cos^2(\theta)|0\kt^{\otimes n-1}\br 0|^{\otimes n-1}  + pe^{\textit{i}\phi}\cos(\theta)\sin(\theta)|0\kt^{\otimes n-1}\br 1|^{\otimes n-1} \nonumber \\ && + pe^{-\textit{i}\phi}\cos(\theta)\sin(\theta)|1\kt^{\otimes n-1}\br 0|^{\otimes n-1} + p \sin^2(\theta)|1\kt^{\otimes n-1}\br 1|^{\otimes n-1}\nonumber \\ && + \sum_{i=0}^{2^{n-1}-1}\left(\frac{1-p}{2^{n-1}}\right)|i\kt\br i|. \label{n_qubit_ra1}
\end{eqnarray}
Let us define $2^{n-1} = L$. The matrix representation of the reduced density matrix $\rho_{A|\Pi_1}$ becomes,
\begin{equation}
\rho_{A|\Pi_1} = \begin{bmatrix}
b_{11}&0&0&\dots&b_{1 L}\\
0&b_{22}&0&\dots& 0\\
0&0&b_{33}&\dots& 0\\
\vdots& & &\ddots&\vdots\\
b_{L 1}& 0 &0 &\dots & b_{LL}
\end{bmatrix}_{LL}, \label{n_qubit_ra1_matrix}
\end{equation}
the corner matrix elements of the above matrix are,
\begin{align}
b_{11} &= p\cos^2{\theta} + \left(\frac{1-p}{2^{n-1}}\right); & b_{L L} = p\sin^2{\theta} + \left(\frac{1-p}{2^{n-1}}\right);\\
b_{1 L} &= pe^{\textit{i}\phi}\cos(\theta)\sin(\theta); &
b_{L 1}= pe^{-\textit{i}\phi}\cos(\theta)\sin(\theta),
\end{align}
and remaining matrix elements are written below,
\begin{align}
b_{22} = b_{33} =\ldots= b_{L-1\;L-1} = \left(\frac{1-p}{2^{n-1}}\right).
\end{align}
The characteristic equation for the matrix given in Eq.(\ref{n_qubit_ra1_matrix}) is shown below,
\begin{eqnarray}
(b_{22}-\beta)(b_{33} - \beta)\cdots(b_{L-1\;L-1}-\beta)\big(\beta^2 - \beta(b_{11}+ b_{LL}) + b_{11}b_{LL}\nonumber \\- b_{L1}b_{1L}\big) = 0. \label{n_qubit_character2}
\end{eqnarray}
The eigenvalues of the matrix in Eq.(\ref{n_qubit_ra1_matrix}), based upon the Eq.(\ref{n_qubit_character2}), can be calculated as,
\begin{align}
\beta_1 = \beta_2 =\ldots=\beta_{L-2} = \left(\frac{1-p}{2^{n-1}}\right),
\end{align}
and remaining two eigenvalues are,
\begin{align}
\beta_{L-1} = \left(\frac{1+(2^{n-1}-1)p}{2^{n-1}}\right) && \text{and} && \beta_L = \left(\frac{1-p}{2^{n-1}}\right).
\end{align}
Now we will calculate the von Neumann entropy of the state defined in Eq.(\ref{n_qubit_ra1}) as it is an important requirement to calculate the quantum conditional entropy given by the Eq.(\ref{quant_cond_ent}),
\begin{eqnarray}
S(\rho_{A|\Pi_1}) &=&  - (2^{n-1} - 2) \left(\frac{1-p}{2^{n-1}}\right)\log_2 \left(\frac{1-p}{2^{n-1}}\right) -\left(\frac{1+(2^{n-1} -1)p}{2^{n-1}}\right) \nonumber \\ && \log_2 \left(\frac{1+(2^{n-1} -1)p}{2^{n-1}}\right) -  \left(\frac{1-p}{2^{n-1}}\right)\log_2 \left(\frac{1-p}{2^{n-1}}\right),
\end{eqnarray}
which can be further simplified into the following equation,
\begin{eqnarray}
 S(\rho_{A|\Pi_1}) &=&- (2^{n-1} - 1) \left(\frac{1-p}{2^{n-1}}\right)\log_2 \left(\frac{1-p}{2^{n-1}}\right) -\left(\frac{1+(2^{n-1} -1)p}{2^{n-1}}\right) \nonumber \\ && \log_2 \left(\frac{1+(2^{n-1} -1)p}{2^{n-1}}\right). \label{n_qubit_sa1}
\end{eqnarray}%
The post measurement state for the subsystem $A$ after the measurement is performed on subsystem $B$ using measurement operator $\Pi_2$ is $\rho_{A|\Pi_2}$, and it has the same matrix structure as that of Eq.(\ref{n_qubit_ra1_matrix}). The associated probability for the state being $\rho_{A|\Pi_2}$ is $p_2$ which is $1/2$. From the above discussion, the quantum conditional entropy can be written as,
\begin{equation}
S(A|B) = S(\rho_{A|\Pi_1}). \label{cond_entropy}
\end{equation}
Note that the expression for quantum conditional entropy in Eq.(\ref{cond_entropy}), is independent of the variables $\theta$ and $\phi$, which are present in the general mathematical expression of the measurement operators in Eq.(\ref{measure_op}). It leads to the fact that the minimization of the quantum conditional entropy over the measurement operators is not required. Using Eqs.(\ref{n_qubit_sb}, \ref{n_qubit_sab}, \ref{cond_entropy}) in Eq.(\ref{discord_final_1}), the expression for the QD finally boils down to,
\begin{eqnarray}
\mathcal{D}(A:B) &=&  1 + (2^n - 1) \left(\frac{1-p}{2^n}\right)\log_2 \left(\frac{1-p}{2^n}\right)+\left(\frac{1+(2^n -1)p}{2^n}\right) \nonumber \\ && \log_2 \left(\frac{1+(2^n -1)p}{2^n}\right) - (2^{n-1} - 1) \left(\frac{1-p}{2^{n-1}}\right)\log_2 \left(\frac{1-p}{2^{n-1}}\right)\nonumber \\ && -\left(\frac{1+(2^{n-1} -1)p}{2^{n-1}}\right) \log_2 \left(\frac{1+(2^{n-1} -1)p}{2^{n-1}}\right). \label{main_qd_eq}
\end{eqnarray}%
Note that Eq.(\ref{main_qd_eq}) is dependent of probability $p$ and the total number of qubits $n$ which also appears in Eq.(\ref{n_qubit_werner1}). We study the variation of QD with $p$ for different number of qubits $n$. To this end, we plot the variation of QD against $p$ for different value of $n$ as shown in Fig.[\ref{main_graph}]. We have plotted the numerical result as inset $(a)$ in Fig.[\ref{main_graph}], we were able to calculate the QD till $n=12$ from a normal computer numerically, the main plot in Fig.[\ref{main_graph}] is an analytical plot giving insights for higher values of $n$ based on Eq.(\ref{main_qd_eq}). It is also to be noted that the analytical plot and numerical plot matches well for all the values of $n$ (till $n=12$) with a precision of the order of $10^{-6}$. For $n=2$, the variation between QD and $p$ exactly matches the result for the two-qubit Werner state as discussed in ref.\cite{zurek} as a self consistent check. From Fig.[\ref{main_graph}] it is evident that for a given value of $p$ the QD increases as $n$ increases, which indicates that the non-local correlations are dependent upon $n$ as given in Eq.(\ref{main_qd_eq}). However, the area under the curve showing the variation between QD and $p$ saturates for large values of $n$ (like $n\geq20$), non-trivial analytical calculations show that for large values of $n$, typically $2^n \rightarrow \infty$, this curve approaches a straight line with unit slope, a phenomenon which we call as the ``saturation of quantum discord". The proof of the aforesaid statement has been completely derived in the subsequent section. Another important observation from Fig.[1] is that the QD is a convex function of the mixing probability $p$ for any $n$. A noteworthy point is that, as the value of $n$ increases and approaches the thermodynamic limit, it becomes the convex roof at saturation, which is nothing but the straight line with unit slope. These observations are proved in detail in appendix A.
\begin{figure}[h]
  \includegraphics[width=0.99\textwidth]{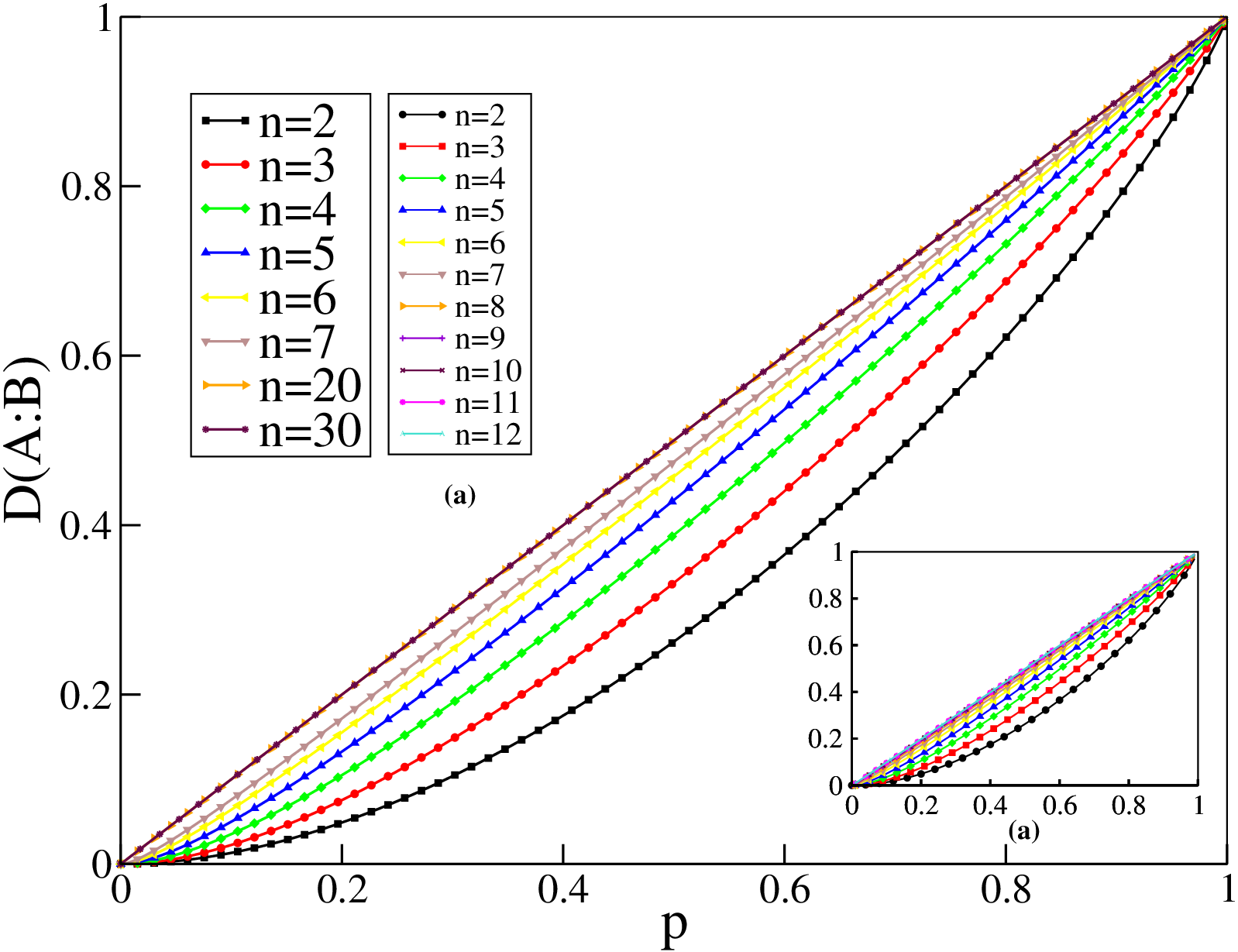}
\caption{Variation of the $\mathcal{D}(A:B)$ with respect to $p$ for different number of qubits $n$. Inset $(a)$ shows the numerical plot between the $\mathcal{D}(A:B)$ versus $p$ for different values of $n$ till $n=12$. Legend box $(a)$ corresponds to the inset plot $(a)$. It is seen that as $n$ increases, the numerical plot tends towards the analytical plot.}
\label{main_graph}
\end{figure}%
\section{The value of quantum discord at the thermodynamic limit (of $n\rightarrow \infty$)} \label{N_QD}
For a generalized $n$-qubit Werner state, the quantum discord, $\mathcal{D}(A:B)$, between any single qubit and the remaining $(n-1)$ qubits was calculated in Eq.(\ref{main_qd_eq}). Now we can state the following lemma about the nature of QD for higher values of $n$.
\begin{lemma}
	For large values of $n$ (as $n\rightarrow \infty$, that is, $2^n\rightarrow \infty$), which is the thermodynamic limit, there exists a linear relationship between the QD and the mixing probability $p$ in the generalized $n$-qubit Werner state.
\end{lemma}
\begin{proof}
The mathematical expression for the QD in Eq.(\ref{main_qd_eq}) can be rewritten as,
\begin{eqnarray}
\mathcal{D}(A:B) &=& 1+\left(\frac{(1-p)(2^n-1)}{2^n}\right)\log_2\left(\frac{1-p}{2^n}\right)+\left(\frac{1+(2^n-1)p}{2^n}\right) \nonumber\\ &&\log_2\left(\frac{1+(2^n-1)p}{2^n}\right)-\left(\frac{(1-p)(2^{n-1}-1)}{2^{n-1}}\right)\log_2\left(\frac{1-p}{2^{n-1}}\right)\nonumber\\ &&-\left(\frac{1+(2^{n-1}-1)p}{2^{n-1}}\right)\log_2\left(\frac{1+(2^{n-1}-1)p}{2^{n-1}}\right),
\end{eqnarray}
after some straightforward simplifications we have,
\begin{eqnarray}
\mathcal{D}(A:B) &=&1+(1-p)\left(1-\frac{1}{2^n}\right)\log_2\left(\frac{1-p}{2^n}\right)+\left(\frac{1}{2^n}+ \left(1-\frac{1}{2^n} \right)p  \right) \nonumber\\ &&\log_2\left(\frac{1}{2^n}+ \left(1-\frac{1}{2^n} \right)p  \right)-(1-p)\left(1-\frac{1}{2^{n-1}} \right)\log_2\left(\frac{1-p}{2^{n-1}} \right)\nonumber\\ &&-\bigg(\frac{1}{2^{n-1}}+ \left(1-\frac{1}{2^{n-1}}\right)  p\bigg)\log_2\left(\frac{1}{2^{n-1}}+ \left(1-\frac{1}{2^{n-1}}\right)p\right).
\end{eqnarray}%
Now consider the limiting case for large values of $n$, that is, $2^n \rightarrow\infty$, the above equation transforms into the following,
\begin{eqnarray}
\lim_{2^n\to\infty}\mathcal{D}(A:B) &=& 1+(1-p)\log_2\left(\frac{1-p}{2^n} \right)+p\log_2(p)-(1-p)\log_2\left(\frac{1-p}{2^{n-1}} \right)\nonumber\\ && -p\log_2(p),\\
&=& 1+(1-p)\log_2\left(\frac{1-p}{2^n} \right)-(1-p)\log_2\left(\frac{1-p}{2^{n-1}}\right),\\
&=& 1+(1-p)\log_2\left(\frac{1}{2} \right),
\end{eqnarray}
which finally gives,
\begin{eqnarray}
\lim_{2^n\to\infty}\mathcal{D}(A:B) &=& p. \label{lin_D_p}
\end{eqnarray}
It is evident from Eq.(\ref{lin_D_p}) that there exists a linear relationship between $\mathcal{D}(A:B)$ and $p$, when $2^n\to\infty$. Now it is trivial to verify that the slope of the line as given in Eq.(\ref{lin_D_p}) is unity, which means that,
\begin{equation}
\lim_{2^n\to\infty}\frac{\partial\mathcal{D}(A:B)}{\partial p}=1.
\end{equation}
Thus our claim is proved.
\end{proof}%
\section{Logarithmic negativity in the $n$-qubit Werner state}
\label{negv}
Till now we understood how non-local correlations varies with increasing number of qubits  $n$ and what really happens to QD as a function of $p$ in the thermodynamic limit of the number of qubits. Now we analyze more deeply into how the entanglement content varies between any single qubit and the remaining $n-1$ qubits for the state as given in Eq.(\ref{n_qubit_werner1}). However, there is no entanglement measure reported till now for a mixed state having total number of qubits $n>2$ which has a closed form expression. Therefore, we use logarithmic negativity \cite{plenio}, which is an entanglement monotone, and it gives an indicator of the entanglement content in the state. To this end, we study the variation of logarithmic negativity with respect to $p$, for different values of $n$. We perform the partial transpose on the subsystem $B$ of the density matrix in Eq.(\ref{n_qubit_werner1}), after some calculations, the final matrix after the partial transposition operation can be written as,
\begin{eqnarray}
\rho_{W_{AB}}^{\Gamma_B} &=& \frac{p}{2}\left[\right.|1\kt \otimes |0\kt^{\otimes n-1}\br 0|  \otimes \br 1|^{\otimes n-1}+ |0\kt \otimes |1\kt^{\otimes n-1}\br 1| \otimes \br 0|^{\otimes n-1} + |0\kt^{\otimes n}\br 0|^{\otimes n} \nonumber \\ && +|1\kt^{\otimes n}\br 1|^{\otimes n}\left.\right] +\sum_{i=0}^{2^n -1}\left(\frac{1-p}{2^n}\right) |i\kt\br i|. \label{partial_eq_1}
\end{eqnarray}
The general structure of the matrix defined in Eq.(\ref{partial_eq_1}) is equivalent to the matrix defined below,
\begin{equation}
\rho_{W_{AB}}^{\Gamma_B} = \begin{bmatrix}
c_{11}&0&0&\dots&\dots&0\\
0&c_{22}&0&\dots& \dots&0\\
\vdots& &\ddots & & &\vdots\\
0&0&\dots&c_{\frac{N}{2}\frac{N}{2}}&c_{\frac{N}{2}\frac{N}{2}+1}\dots& 0\\
0&0&\dots&c_{\frac{N}{2}+1\frac{N}{2}}&c_{\frac{N}{2}+1\frac{N}{2}+1}\dots& 0\\
\vdots& & &&\ddots &\vdots\\
0& & & & &c_{N N}
\end{bmatrix}_{N \times N}. \label{partial_neg_mat}
\end{equation}
The matrix in Eq.(\ref{partial_neg_mat}) has all the elements are zero except for the diagonal elements, and two off diagonal elements, which are written as follows, $c_{\frac{N}{2}\frac{N}{2}+1}$ and $c_{\frac{N}{2}+1\frac{N}{2}}$. The non-zero matrix elements can be written as follows,
\begin{eqnarray}
c_{11} &=& c_{NN} = \left(\frac{1+\left(2^{n-1}-1\right)p}{2^n}\right);\\
c_{\frac{N}{2}\frac{N}{2}+1} &=& c_{\frac{N}{2}+1\frac{N}{2}} = \frac{p}{2};\\
c_{22}&=&\ldots=c_{N-1N-1} = \left(\frac{1-p}{2^n}\right). \label{par_ele_def}
\end{eqnarray}
The characteristic equation for the matrix defined in Eq.(\ref{partial_neg_mat}) can be written in a factorizable form as,
\begin{eqnarray}
(c_{22}-\ld')(c_{33} - \ld')\cdots(c_{N-1\;N-1}-\ld')\left(\ld'^2 - \ld'(c_{11}+ c_{NN}) + c_{11}c_{NN} \right. \nonumber \\ \left.- c_{\frac{N}{2}\frac{N}{2}+1}c_{\frac{N}{2}+1 \frac{N}{2}}\right) = 0. \label{c_char}
\end{eqnarray}
The roots (eigenvalues) of the characteristic equation in Eq.(\ref{c_char}) are given below,
\begin{eqnarray}
\ld'_1 &=& \left(\frac{1-p}{2^n}\right) - \frac{p}{2}; \label{partial_eigen1}\\
\ld'_2 &=&\ld'_{N-1} = \ld'_{N}= \left(\frac{1-p}{2^n}\right) + \frac{p}{2}; \label{partial_eigen2}\\
\ld'_3 &=& \ld'_4 = \cdots\cdots =\ld'_{N -2} = \left(\frac{1-p}{2^n}\right). \label{partial_eigen}
\end{eqnarray}
All the obtained eigenvalues in Eqs.(\ref{partial_eigen1} - \ref{partial_eigen}) are positive for any value of $p$ except the eigenvalue $\ld_1'$. In ref.\cite{peres}, the author has shown that the necessary criteria for a density matrix to be separable is to have non-negative eigenvalues of the partially transposed density matrix. However, the subsystems $A$ and $B$ in the generalized $n$-qubit Werner state are inseparable for a range of mixing probabilities $p$, when the eigenvalue $\ld_1'$ is negative. The range of $p$ where $\ld_1'$ is negative is shown below,
\begin{equation}
\left(\frac{1-p}{2^n}\right) - \frac{p}{2} < 0,
\end{equation}
which boils down to,
\begin{equation}
p> \frac{1}{1+ 2^{n-1}}. \label{condition1}
\end{equation}
For $n=2$ the above equation transform into $p>1/3$, which is consistent with the concurrence measure of the two-qubit Werner state \cite{wootters},\cite{dimer}. The value of concurrence of two-qubit Werner state is zero  (separable state) in the range $0\leq p<1/3$, and non-zero (inseparable state) in the range $1/3\leq p\leq 1$. From Eq.(\ref{condition1}), we observe that the range of $p$ for which the state is separable decreases as $n$ increases. To verify Eq.(\ref{condition1}) with respect to state given in Eq.(\ref{n_qubit_werner1}), we calculate the logarithmic negativity between the subsystems $A$ and $B$ of the state. As a first step we calculate the trace norm of the matrix $\rho_{W_{AB}}^{\Gamma_B}$ as shown below,
\begin{eqnarray}
||\rho_{W_{AB}}^{\Gamma_B}|| &=& Tr\left(\sqrt{\left(\rho_{W_{AB}}^{\Gamma_B}\right)^{\dagger}\rho_{W_{AB}}^{\Gamma_B}}\right) = Tr\left(\sqrt{(\rho_{W_{AB}}^{\Gamma_B})^2}\right),\label{tr_norm}
\end{eqnarray}
using the fact that the trace of the matrix is the sum of its eigenvalues in Eqs.(\ref{partial_eigen1} - \ref{partial_eigen}), we can rewrite Eq.(\ref{tr_norm}) as,
\begin{equation}
||\rho_{W_{AB}}^{\Gamma_B}|| = \sum_{j=1}^{2^n}\sqrt{(\ld_j')^2}. \label{tr_norm2}
\end{equation}
Now the logarithmic negativity can be computed as,
\begin{equation}
\mathcal{N}_L = \log_2\left(||\rho_{W_{AB}}^{\Gamma_B}||\right)=\log_2\bigg(\sum_{j=1}^{2^n}\sqrt{(\ld_j')^2} \bigg). \label{log_neg}
\end{equation}
For the interval $0\leq p \leq 1/(1+ 2^{n-1})$, the eigenvalue $\lambda'_1$ is positive, therefore the value of the trace norm in Eq.(\ref{tr_norm2}) will be unity, and by the definition of logarithmic negativity ($\mathcal{N}_L$) in Eq.(\ref{log_neg}), it will be zero, this indicates that the state is unentangled in the above interval. However, in the interval $1/(1+ 2^{n-1}) < p \leq 1$, the trace norm as defined in Eq.(\ref{tr_norm2}) reduces to,
\begin{equation}
||\rho_{W_{AB}}^{\Gamma_B}||	= 1 - 2\ld_1 = \frac{(2^{n-1} +1)p + (2^{n-1}-1)}{2^{n-1}}.\label{new_tr_norm}
\end{equation}
Based on Eq.(\ref{new_tr_norm}), the logarithmic negativity becomes,
\begin{equation}
\mathcal{N}_L = \log_2\left(\frac{(2^{n-1} +1)p + (2^{n-1}-1)}{2^{n-1}}\right). \label{final_nl}
\end{equation}
From Eq.(\ref{final_nl}) we observe that $\mathcal{N}_L$ in the interval $1/(1+ 2^{n-1}) < p \leq 1$ is non-zero (inseparable), and is dependent on both $n$ and $p$. In Fig.[\ref{log_neg_fig}], we have plotted $\mathcal{N}_L$ versus $p$ for different values of $n$, and we observe that the range of $p$ for which the state is inseparable increases, not only that, even the entanglement content also increases as we increase $n$. The increment in the range of $p$ for which the state remains inseparable is non-uniform which is basically a difference between $\frac{1}{1+2^{k-1}}$ and $\frac{1}{1+2^{k}}$ where, $2\leq k \leq (n-1)$. However, for very large values of $n$ ($2^n\rightarrow\infty$), it is easy to see that, the logarithmic negativity as given in Eq.(\ref{final_nl}) reduces to a simple expression,
\begin{equation}
\lim_{2^n\to\infty} \mathcal{N}_L = \log_2(1+p). \label{long_range_nl}
\end{equation}
\begin{figure*}[h]
  \includegraphics[width=0.95\textwidth]{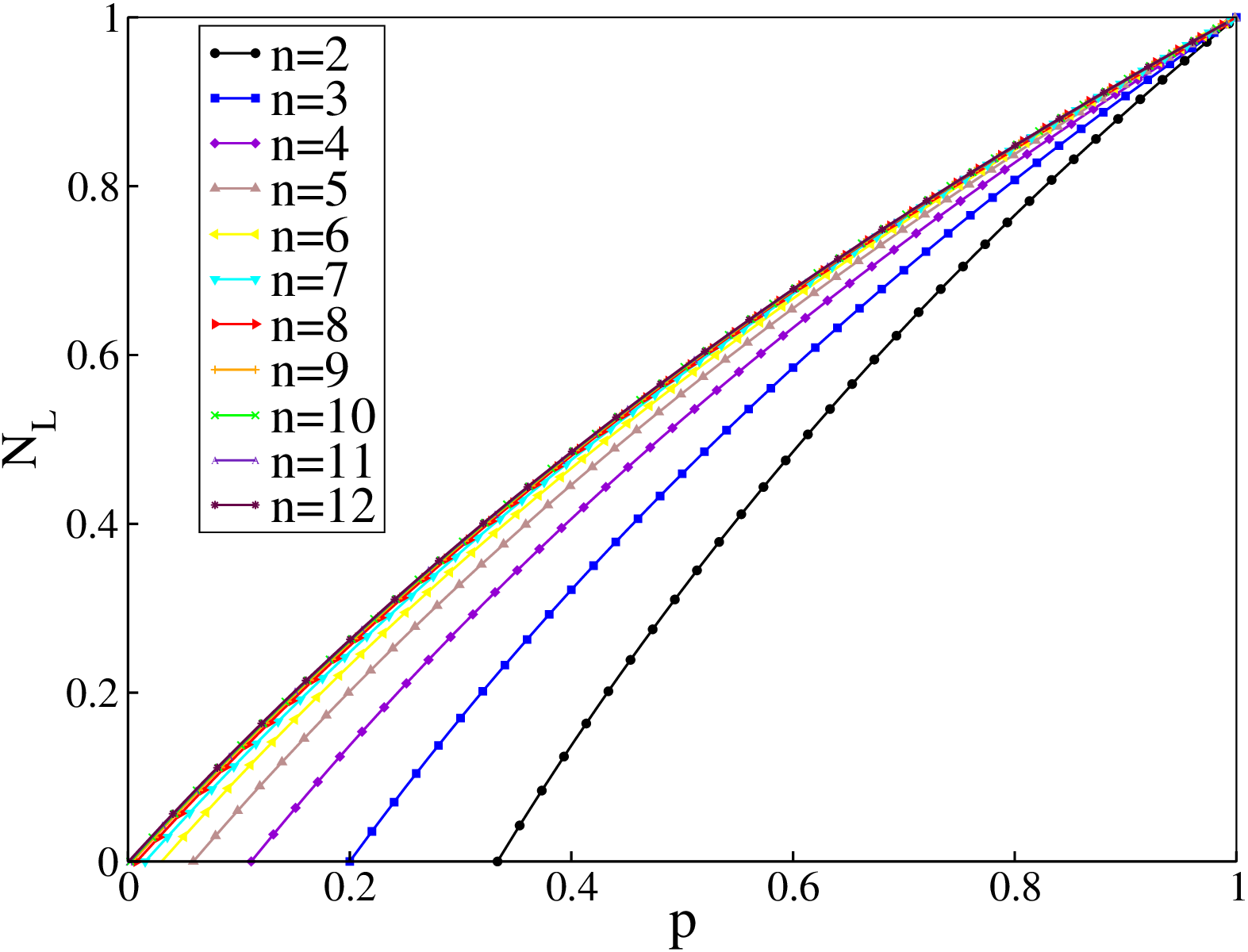}
\caption{Numerical results showing the variation of $\mathcal{N}_L$ with respect to $p$ for different number of qubits $n$ indicating a non-uniform behaviour.}
\label{log_neg_fig}
\end{figure*}%
As we see in Fig.[2], the logarithmic negativity $\mathcal{N}_L$ is a concave function of the mixing probability $p$. A proof of the preceding statement can be found in appendix B. \\
Thus it can be concluded that the increase in the entanglement content saturates for large value of $n$. In Fig.[\ref{log_neg_fig}], we observe that for higher number of qubits, almost for the entire range of $p$ ($0\leq p \leq 1$) the state remains entangled, which concurs with Eq.(\ref{condition1}). We now study the change of logarithmic negativity with respect to the variable $p$, in order to gain more insights on the rate of change in the entanglement content with respect to $p$, the first order derivative of $\mathcal{N}_L$ with respect to $p$ is computed as below, 
\begin{equation}
\frac{\partial \mathcal{N}_L}{\partial p} =\frac{1}{\ln2}\left( \frac{\left(2^{n-1}+1\right)}{\left(2^{n-1}+1\right)p + \left(2^{n-1}-1\right)}\right) = \frac{1}{\ln2\left(p+\dfrac{\left(2^{n-1}-1\right)}{\left(2^{n-1}+1\right)}\right)}. \label{varywith_p}
\end{equation}%
For very large values of $n$, it is simple to see that Eq.(\ref{varywith_p}) transforms into the following,
\begin{equation}
\lim_{2^n\to\infty}\frac{\partial \mathcal{N}_L}{\partial p}=\frac{1}{\ln2}\left(\frac{1}{1+p}\right),\label{large_varywith_p}
\end{equation}%
which is the first order derivative of Eq.(\ref{long_range_nl}) with respect to $p$. From Eq.(\ref{large_varywith_p}), we can conclude that for large values of $n$, the entanglement content in the state does not rise sharply as compared to the smaller values of $n$ in the system, which is also evident from Fig.[\ref{log_neg_fig}].
\section{Conclusion}
\label{conclusion}
In this paper we have constructed a generalized $n$-qubit Werner state, and studied it in extensive detail with regard to QD and logarithmic negativity for any single qubit ($B$ subsystem) and the remaining $n-1$ qubits ($A$ subsystem) of a bipartite split of the state. We have proved that as the total number of qubits ($n$) increases, the QD as a function of $p$ saturates to a straight line and in fact, in the thermodynamic limit of the number of qubits tending to infinity, the QD is a straight line with unit slope, and this observation is backed by our numerical results also. The convexity property of QD with respect to the mixing probability $p$ has been discussed in detail with relevant analytical and numerical computations. We have then extended our discussion to the variation of entanglement content between the subsystems $A$ and $B$ using logarithmic negativity, subjected to the changes in the parameters $p$ (mixing probability) and $n$ (total number of qubits). Analytically as well as numerically it has been proved that, the range of $p$ for which the state is inseparable, increases. The entanglement content also increase as we increase $n$. The concave nature of logarithmic negativity with respect to the mixing probability $p$ has been observed and studied in detail. We like to point out that this paper discusses the non-local correlations between a single qubit and the remaining $n-1$ qubits. However, it is an open question to extend this for a bipartite system containing an arbitrary number of qubits in each of its parts, for the state under consideration. This is due to the fact that the complexities in calculating QD in such cases dramatically increases due to exponential increase in the parameters needed to be minimized as discussed earlier in the introduction.
\begin{acknowledgements}
We would like to thank Ms. Payal D. Solanki for useful discussions.
\end{acknowledgements}
%
\begin{appendices}
	\section{Convexity of Quantum Discord}
	\renewcommand{\theequation}{A.\arabic{equation}}
	\setcounter{equation}{0}
	It can be seen from Fig.[1] ($\mathcal{D}(A:B)$ versus $p$) that the quantum discord is a convex function over the range of mixing probability $p$. We can prove analytically as well as numerically that the expression for quantum discord written in Eq.(\ref{eq.1}) below is a convex function for $0<p<1$. To this end, we will use results from differential calculus which says that the function $\mathcal{D}(A:B)$ is a convex function in the aforesaid range of $p$, if $\frac{\partial^2 \mathcal{D}(A:B)}{\partial p^2} > 0$ for all values of $n$.
	We start by writing the analytical expression of $\mathcal{D}(A:B)$ as,
	\begin{eqnarray}
		\mathcal{D}(A:B) &=&  1 + (2^n - 1) \left(\frac{1-p}{2^n}\right)\log_2 \left(\frac{1-p}{2^n}\right)+\left(\frac{1+(2^n -1)p}{2^n}\right) \nonumber \\ && \log_2 \left(\frac{1+(2^n -1)p}{2^n}\right) - (2^{n-1} - 1) \left(\frac{1-p}{2^{n-1}}\right)\log_2 \left(\frac{1-p}{2^{n-1}}\right) \nonumber \\ &&-\left(\frac{1+(2^{n-1} -1)p}{2^{n-1}}\right) \log_2 \left(\frac{1+(2^{n-1} -1)p}{2^{n-1}}\right). \label{qd_eq} \label{eq.1}
	\end{eqnarray}
	Now we will rearrange the above equation in terms of new variables $x$, $y$ and $z$ as shown below,
	\begin{eqnarray}
		x&=& 1-p, \label{eq.2}\\
		y&=& 1+ (2^n -1)p,\\
		z&=& 1 + (2^{n-1}-1)p. \label{eq.4}
	\end{eqnarray}
	Based on Eqs.(\ref{eq.2}-\ref{eq.4}), we have the following relations between $x$, $y$ and $z$.
	\begin{eqnarray}
		2z &=& x+y, \label{eq.5}\\
		y &=& x+ 2^np. \label{eq.6}
	\end{eqnarray}
	Now, $\mathcal{D}(A:B)$ can be written in terms of $x$, $y$ and $z$ as,
	\begin{eqnarray}
		\mathcal{D}(A:B) &=& 1 + (2^n -1)\frac{x}{2^n}\log_2\left(\frac{x}{2^n}\right) + \frac{y}{2^n}\log_2\left(\frac{y}{2^n}\right) - (2^{n-1}-1) \nonumber \\ &&\frac{x}{2^{n-1}}\log_2 \left(\frac{x}{2^{n-1}}\right) - \frac{z}{2^{n-1}}\log_2\left(\frac{z}{2^{n-1}}\right).\label{eq7}
	\end{eqnarray}
	We can use the property, $\log(1/u) = -\log(u)$ (for $u>0$) and $\log_22^n = n$ for further simplification of Eq.(\ref{eq7}) to get,
	\begin{eqnarray}
		\mathcal{D}(A:B) &=&1 + \left(\frac{2^n -1}{2^n}\right) x \left(\log_2 x - n\right) + \left(\frac{y}{2^n}\right)(\log_2 y -n) - \left(\frac{2^{n-1}-1}{2^{n-1}}\right)\nonumber \\ &&x(\log_2 x - (n-1))  - \left(\frac{z}{2^{n-1}}\right)(\log_2 z - (n-1)), \\
		&=& 1+\left(\frac{2^n -1 }{2^n} - \frac{2^{n-1}-1}{2^{n-1}}\right)x\log_2 x - \frac{xn(2^n - 1)}{2^n} + \left(\frac{y}{2^n}\right)\log_2 y \nonumber \\ &&- \frac{ny}{2^n} +(n-1)\left(\frac{2^{n-1}-1}{2^{n-1}}\right)x - \left(\frac{z}{2^{n-1}}\right)\log_2 z + \left(\frac{n-1}{2^{n-1}}\right)z.
	\end{eqnarray}
	The above equation can be further rearranged as follows,
	\begin{eqnarray}
		\mathcal{D}(A:B) &=& 1+ x\left[\frac{-n(2^n-1)}{2^n} + \frac{(n-1)(2^{n-1}-1)}{2^{n-1}} \right] -\frac{ny}{2^n} + \left(\frac{n-1}{2^n}\right)\nonumber \\ &&(2z) + \left( \frac{x}{2^n}\right)\log_2x + \left(\frac{y}{2^n}\right)\log_2 y - \left(\frac{z}{2^{n-1}}\right)\log_2z. \label{eq.10}
	\end{eqnarray}
	Using Eq.(\ref{eq.5}), we can substitute for $2z$ in terms of $x$ and $y$ in Eq.(\ref{eq.10}) to get,
	\begin{eqnarray}
		\mathcal{D}(A:B) &=& 1+ x\left[\frac{-n(2^n-1)}{2^n} + \frac{(n-1)(2^{n-1}-1)}{2^{n-1}} \right] -\frac{ny}{2^n} + \left(\frac{n-1}{2^n}\right)\nonumber \\ &&(x+y) + \left( \frac{x}{2^n}\right)\log_2x + \left(\frac{y}{2^n}\right)\log_2 y - \left(\frac{z}{2^{n-1}}\right)\log_2z, \\
		&=&1+ x\left[\frac{-n(2^n-1)}{2^n} + \frac{(n-1)(2^n-2)}{2^n} + \frac{(n-1)}{2^n}\right] -\frac{ny}{2^n} \nonumber \\ &&+ \frac{ny}{2^n} - \frac{y}{2^n} + \left( \frac{x}{2^n}\right)\log_2x + \left(\frac{y}{2^n}\right)\log_2 y - \left(\frac{z}{2^{n-1}}\right)\log_2z.\label{eq.12}
	\end{eqnarray}
	Again using Eq.(\ref{eq.6}), we can write $y$ in terms of $x$ in Eq.(\ref{eq.12}) and therefore, the above equation becomes,
	\begin{eqnarray}
		\mathcal{D}(A:B) &=& x\left[\frac{-n(2^n-1)}{2^n} + \frac{(n-1)(2^n-2)}{2^n} + \frac{(n-1)}{2^n}\right]- \frac{(x + 2^np)}{2^n} + 1\nonumber \\ && + \left( \frac{x}{2^n}\right)\log_2 x + \left(\frac{y}{2^n}\right)\log_2 y - \left(\frac{z}{2^{n-1}}\right)\log_2 z,\\
		&=&x\left[\frac{-n(2^n-1)}{2^n} + \frac{(n-1)(2^n-2)}{2^n} + \frac{(n-1)}{2^n}\right]- \frac{x}{2^n} - p + 1\nonumber \\ && + \left( \frac{x}{2^n}\right)\log_2 x + \left(\frac{y}{2^n}\right)\log_2 y - \left(\frac{z}{2^{n-1}}\right)\log_2 z,\\
		&=&x\left[\frac{-n(2^n-1)}{2^n} + \frac{(n-1)(2^n-2)}{2^n} + \frac{(n-1)}{2^n} - \frac{1}{2^n} + 1\right]\nonumber \\ && + \left( \frac{x}{2^n}\right)\log_2 x + \left(\frac{y}{2^n}\right)\log_2 y - \left(\frac{z}{2^{n-1}}\right)\log_2z. \label{eq.15}
	\end{eqnarray}
	It can be observed from Eq.(\ref{eq.15}) that the coefficient of $x$ is independent of $p$ and it is equal to zero for any value of $n$. Therefore, we can write the simplified expression for $\mathcal{D}(A:B)$ as follows,
	\begin{eqnarray}
		\mathcal{D}(A:B) = \frac{1}{2^n}\left\{x\log_2 x + y\log_2 y - 2z\log_2 z \right\}. \label{eq.16}
	\end{eqnarray}
	The first order derivative of $\mathcal{D}(A:B)$ in Eq.(\ref{eq.16}) with respect to $p$ is calculated to be,
	\begin{eqnarray}
		\frac{\partial \mathcal{D}(A:B)}{\partial p} &=& \frac{1}{2^n}\left\{x\times\frac{1}{x\ln2}\left(\frac{\partial x}{\partial p}\right) + \left(\frac{\partial x}{\partial p}\right)\log_2x + y\times\frac{1}{y\ln2}\left(\frac{\partial y}{\partial p} \right)+\right.\nonumber \\ &&\left. \left(\frac{\partial y}{\partial p}\right)\log_2y-2\left( z\times\frac{1}{z\ln2}\left(\frac{\partial z}{\partial p} \right)+\left(\frac{\partial z}{\partial p}\right)\log_2z\right)\right\}, \\
		&=&\frac{1}{2^n}\left\{\left(\frac{\partial x}{\partial p}\right)\left(\frac{1}{\ln2} + \log_2x\right) + \left(\frac{\partial y}{\partial p}\right)\left(\frac{1}{\ln2} + \log_2y\right) \right.\nonumber \\ &&\left.- 2 \left(\frac{\partial z}{\partial p}\right)\left(\frac{1}{\ln2} + \log_2z\right) \right\}. \label{eq.18}
	\end{eqnarray}
	The second order derivative of $\mathcal{D}(A:B)$ can be calculated using Eq.(\ref{eq.18}) as follows,
	\begin{eqnarray}
		\frac{\partial^2 \mathcal{D}(A:B)}{\partial p^2 } &=& \frac{1}{2^n}\left\{ \left(\frac{1}{\ln2}+\log_2x\right)\left(\frac{\partial^2x}{\partial p^2}\right) + \left( \frac{\partial x}{\partial p}\right)\times \frac{1}{x\ln2}\left(\frac{\partial x}{\partial p}\right) +\right. \nonumber \\ && \left. \left(\frac{1}{\ln2}+\log_2y\right)\left(\frac{\partial^2y}{\partial p^2}\right)+\left( \frac{\partial y}{\partial p}\right)\times \frac{1}{y\ln2}\left(\frac{\partial y}{\partial p}\right) -2\right. \nonumber \\ && \left. \left(\frac{1}{\ln2}+\log_2z\right)\left(\frac{\partial^2 z}{\partial p^2}\right) -2 \left( \frac{\partial z}{\partial p}\right)\times \frac{1}{z\ln2}\left(\frac{\partial z}{\partial p}\right) \right\}. \label{eq.19}
	\end{eqnarray}
	Recall Eqs.(\ref{eq.2}-\ref{eq.4}), we can compute the first and second order derivatives of $x$, $y$ and $z$ with respect to $p$ occurring  in Eq.(\ref{eq.19}) as,
	\begin{align}
		\frac{\partial x}{\partial p} = -1 \hspace{1cm}&& \text{and} && \frac{\partial^2 x}{\partial p^2} = 0, \label{eq.20}\\
		\frac{\partial y}{\partial p} = (2^n -1) \hspace{0.2cm}&& \text{and} && \frac{\partial^2 y}{\partial p^2} = 0, \\
		\frac{\partial z}{\partial p} = (2^{n-1}-1) \hspace{-0.16cm}&& \text{and} && \frac{\partial^2 z}{\partial p^2}=0. \label{eq.22}
	\end{align}
	Since all $x$, $y$ and $z$ are linear functions of $p$, it can be observed that all the second order partial derivatives of $x$, $y$ and $z$ in Eq.(\ref{eq.19}) are zero as shown in Eqs.(\ref{eq.20}-\ref{eq.22}). Considering the above facts, Eq.(\ref{eq.19}) can be rewritten as,
	\begin{eqnarray}
		\frac{\partial^2 \mathcal{D}(A:B)}{\partial p^2 } &=& \frac{1}{2^n\ln2}\left\{\frac{1}{x}\left(\frac{\partial x}{\partial p}\right)^2+\frac{1}{y}\left(\frac{\partial y}{\partial p}\right)^2 -\frac{2}{z}\left(\frac{\partial z}{\partial p}\right)^2\right\}. \label{eq.23}
	\end{eqnarray}
	Substituting the values of the first order derivatives of $x$, $y$ and $z$ from Eqs.(\ref{eq.20}-\ref{eq.22}) in Eq.(\ref{eq.23}), we get,
	\begin{eqnarray}
		\hspace{-1cm}\frac{\partial^2 \mathcal{D}(A:B)}{\partial p^2} &=& \frac{1}{2^n\ln2}\left\{\frac{(-1)^2}{(1-p)} + \frac{(2^n -1)^2}{(1+(2^n-1)p)} - \frac{2(2^{n-1}-1)^2}{(1+(2^{n-1}-1)p)}\right\}. \label{eq.24}
	\end{eqnarray}%
	The task now reduces to show that the RHS of Eq.(\ref{eq.24}) is greater than zero. To this effect, the quantity inside the flower brackets in RHS of Eq.(\ref{eq.24}) can be expanded out as shown below,
	\begin{eqnarray}
		\frac{\partial^2 \mathcal{D}(A:B)}{\partial p^2}&=& \frac{\alpha}{\gamma}, \label{eq.25}
	\end{eqnarray}
	where, $\alpha$ the numerator of Eq.(\ref{eq.25}) is given by, 
	\begin{eqnarray}
	\alpha &=& \left[(1+(2^n-1)p)(1+(2^{n-1}-1)p) + (2^n -1)^2(1-p)(1+(2^{n-1}-1)p)\right.\nonumber \\ &&\left.- 2(2^{n-1}-1)^2(1-p)(1+(2^n-1)p)\right],
	\end{eqnarray}
	and $\gamma$, the denominator of Eq.(\ref{eq.25}) is given by, 
	\begin{eqnarray}
	\gamma &=& 2^n\ln2(1-p)(1+(2^n-1)p)(1+(2^{n-1}-1)p).
	\end{eqnarray}
	The denominator $\gamma$ of Eq.(\ref{eq.25}) contains terms, $(1-p), (1+(2^n-1)p)$ and $(1+(2^{n-1}-1)p)$ which are positive in the range of $p$ considered for any $n$, therefore, the denominator is a positive quantity. Since we need to prove that $\frac{\partial^2 \mathcal{D}(A:B)}{\partial p^2}>0$, which implies, we have to show that the numerator $\alpha$ should also be greater than zero. To this end, we can rewrite the numerator of Eq.(\ref{eq.25}) as a quadratic expression in $p$ by grouping the appropriate terms as,
	\begin{align}
		\left[(2^n-1)(2^{n-1}-1) - (2^n -1)^2(2^{n-1}-1) + 2(2^{n-1}-1)^2(2^n-1)\right]p^2 +  \nonumber \\ \left[(2^n -1) + (2^{n-1}-1) - (2^n -1)^2 + (2^n-1)^2(2^{n-1}-1) + 2(2^{n-1}-1)^2 \right. \nonumber \\ \left.- 2(2^{n-1}-1)^2(2^n-1)\right] p + 1 + (2^n -1)^2 - 2(2^{n-1}-1)^2. \label{eq.26}
	\end{align}%
	Here we will first simplify the coefficient of $p^2$, then $p$ and finally the constant term individually. Therefore, considering the coefficient of $p^2$ in Eq.(\ref{eq.26}) as shown below,
	\begin{eqnarray}
		\lefteqn{\left[(2^n-1)(2^{n-1}-1) - (2^n -1)^2(2^{n-1}-1) + 2(2^{n-1}-1)^2(2^n-1)\right]} \nonumber \\
		&=& (2^n-1)(2^{n-1}-1)\left(1-(2^n-1)+2(2^{n-1}-1)\right)\nonumber \\
		&=& (2^n-1)(2^{n-1}-1)\left(1-2^n +1 +2^n -2\right)\nonumber \\
		&=& 0.
	\end{eqnarray}
	Therefore, the coefficient of $p^2$ term is zero, now we consider the coefficient of $p$ in Eq.(\ref{eq.26}) as follows,
	\begin{eqnarray}
		\lefteqn{\left[(2^n -1) + (2^{n-1}-1) - (2^n -1)^2 + (2^n-1)^2(2^{n-1}-1) + 2(2^{n-1}-1)^2 \right.} \nonumber \\ &&\left.- 2(2^{n-1}-1)^2(2^n-1)\right] \nonumber \\
		&=&\left[2^n -1 + 2^{3n-1} + 2^{n-1} - 2^{2n} - 2^{2n+1} - 2 + 2^{n+2} + 2^{n-1} -1 - 2^{3n-1} \right. \nonumber \\ &&\left.- 2^{n+1} + 2^{2n+1}+2^{2n} + 2^2 - 2^{n+2}\right] \nonumber \\
		&=& [2^n + 2.2^{n-1} - 2^{n+1}] = [2.2^n = 2^{n+1}] = [2^{n+1} - 2^{n+1}] \nonumber \\
		&=& 0.
	\end{eqnarray}
	Therefore, the coefficient of $p$ is also zero. Now we consider the constant term in Eq.(\ref{eq.26}), which is,
	\begin{eqnarray}
		\lefteqn{1 + (2^n -1)^2 - 2(2^{n-1}-1)^2} \nonumber \\
		&=& 1 + 2^{2n} +1 - 2^{n+1} - 2^{2n-1} - 2 + 2^{n+1} \nonumber \\
		&=& 2^{2n} - \frac{2^{2n}}{2} = 2^{2n-1}.
	\end{eqnarray}
	Finally, the only term that remains in the numerator of Eq.(\ref{eq.25}) is $2^{2n-1}$, which is positive for the considered range of $p$ as $n\geq2$ and where $n$ is an integer. Hence, we have proved that,
	\begin{eqnarray}
		\frac{\partial^2 \mathcal{D}(A:B)}{\partial p^2} >0.
	\end{eqnarray}%
	Therefore, $\mathcal{D}(A:B)$ is a convex function in the considered range of $p$. The above proof is completely analytical, however, to reinforce our claim and get deeper insights, we have done numerical analysis of the convexity property of $\mathcal{D}(A:B)$ for any $n$ and $p$ in the considered range. To this end, the second order derivative $\frac{\partial^2 \mathcal{D}(A:B)}{\partial p^2}$ was calculated numerically for different values of $n$ in the considered range of $p$, and the plots in Fig.[\ref{convex}] show conclusively that $\mathcal{D}(A:B)$ is indeed a convex function.

	\begin{figure*}[h]
		\centering
		\includegraphics[width = 0.95\textwidth]{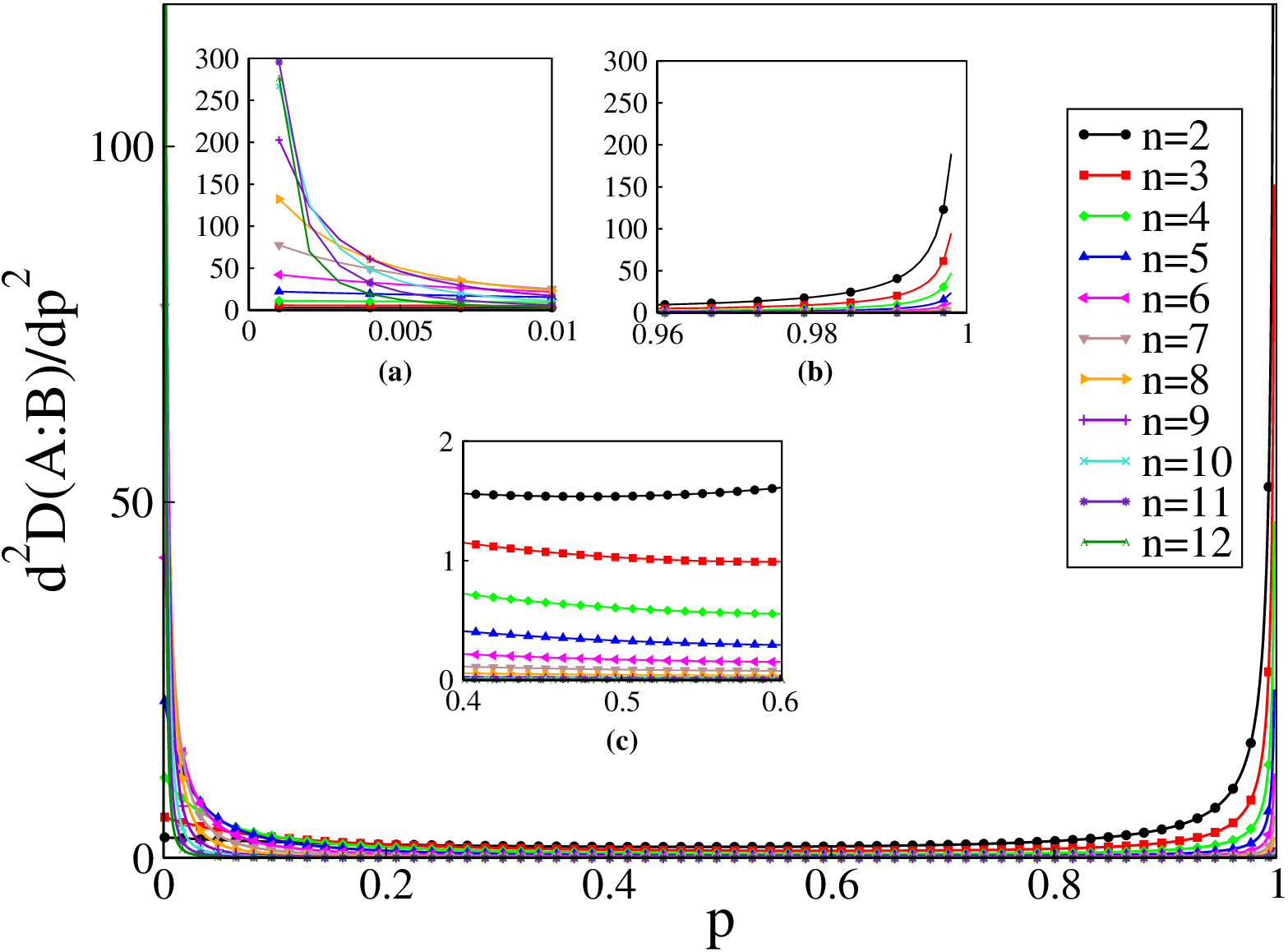}
		\caption{Numerical plot for the variation of the second order derivative of $\mathcal{D}(A:B)$ with respect to $p$ for different values of $n$. The insets show more closely the convexity of $\mathcal{D}(A:B)$ in the split ranges of $p$ as shown.
		} \label{convex}
	\end{figure*}
	\section{Concavity of Logarithmic Negativity}
	\renewcommand{\theequation}{B.\arabic{equation}}
	\setcounter{equation}{0}
Now we will prove that the logarithmic negativity, as given below in Eq.(\ref{eq.31}) is a concave function. Again by using differential calculus in the range of $p$ such that $\frac{1}{1+2^{n-1}}<p<1$, if $\frac{\partial^2 \mathcal{N}_L}{\partial p^2} <0$ for all values of $n$, then $\mathcal{N}_L$ is a concave function.
Here we have logarithmic negativity $\mathcal{N}_L$ dependent upon $p$ in the range $\frac{1}{1+2^{n-1}}<p<1$. Considering the expression for logarithmic negativity in the above range of $p$ we have,
\begin{eqnarray}
	\mathcal{N}_L = \log_2\left(\frac{(2^{n-1}+1)p + (2^{n-1}-1)}{2^{n-1}}\right). \label{eq.31}
\end{eqnarray}
We have to calculate the second order derivative of $\mathcal{N}_L$ with respect to $p$. To this end, we already know the first order derivative of $\mathcal{N}_L$ with respect to $p$ from Eq.(\ref{varywith_p}) as, 
\begin{eqnarray}
	\frac{\partial \mathcal{N}_L}{\partial p} &=&\frac{1}{\ln2\left(p+\left(\frac{2^{n-1}-1}{2^{n-1}+1}\right)\right)}. \label{eq.33}
\end{eqnarray}
We can now differentiate Eq.(\ref{eq.33}) with respect to $p$, to obtain second order derivative of $\mathcal{N}_L$ with respect to $p$ as,
\begin{eqnarray}
	\frac{\partial^2 \mathcal{N}_L}{\partial p^2} &=& \frac{1}{\ln2}\left(\frac{-1}{\left(p + \left(\frac{2^{n-1}-1}{2^{n-1}+1}\right)\right)^2}\right). \label{eq.34}
\end{eqnarray}
Since the denominator of the RHS in Eq.(\ref{eq.34}) is always positive for the considered range of $p$, it is transparent to see that the RHS of Eq.(\ref{eq.34}) is always negative. Therefore, it can be concluded that,
\begin{eqnarray}
	\frac{\partial^2 \mathcal{N}_L}{\partial p^2} <0, \label{eq.35}
\end{eqnarray}
\begin{figure*}[h]
	\centering
	\includegraphics[width = 0.95\textwidth]{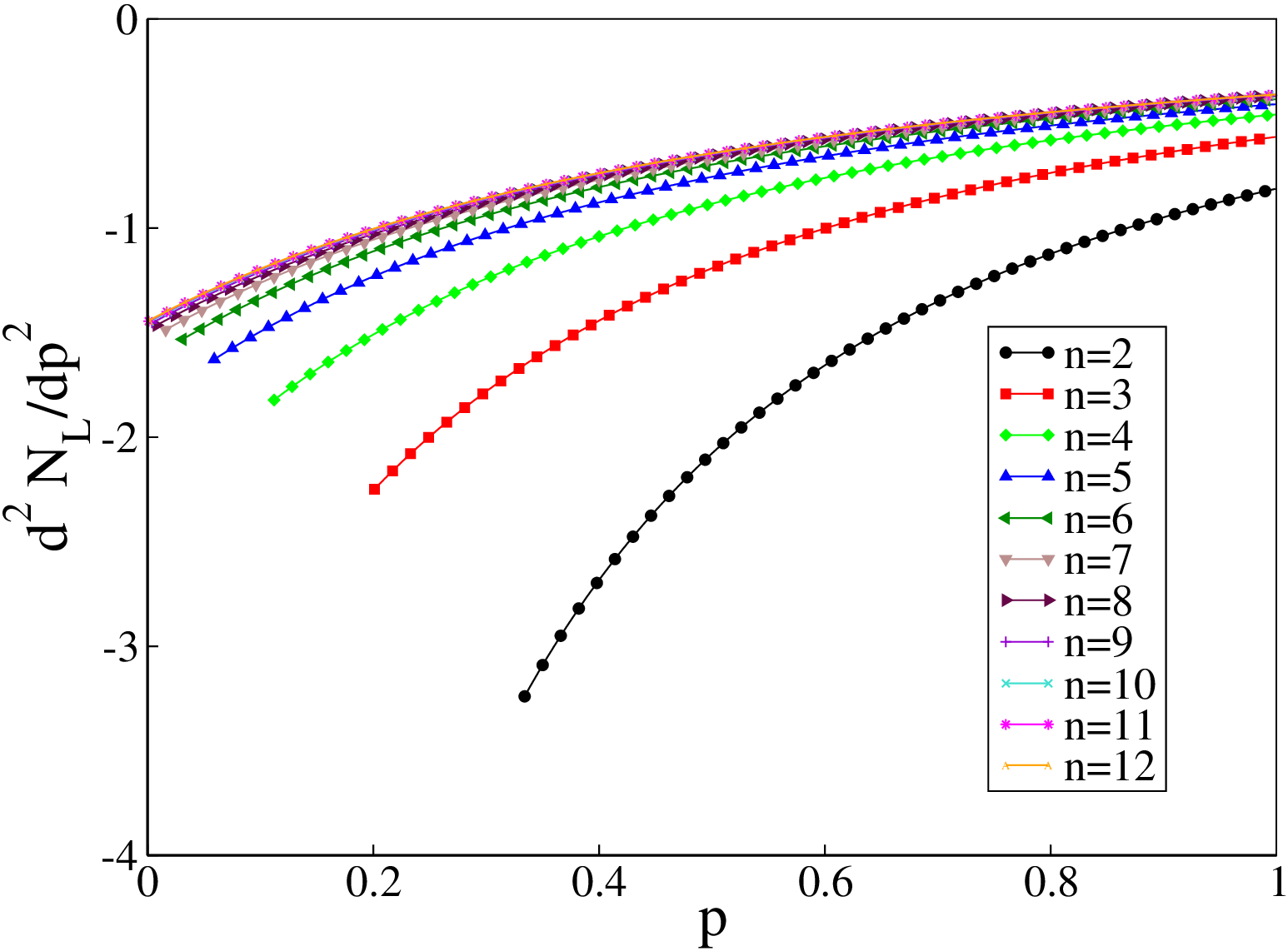}
	\caption{Numerical plot for the variation of the second order derivative of logarithmic negativity ($\mathcal{N}_L$) with respect to $p$ such that $\frac{1}{1+2^{n-1}}<p<1$, for different values of $n$, showing clearly the concavity of $\mathcal{N}_L$.}\label{concave}
\end{figure*}%
for all values of $p$ in the range $\frac{1}{1+2^{n-1}}<p<1$ and for any value of $n$. Thus the logarithmic negativity is proved to be a concave function. The above result is analytical, we also calculated numerically the second order derivative of Eq.(\ref{eq.31}) with respect to $p$, to verify our analytical results. We have plotted the numerical results in Fig.[\ref{concave}] above, from which it can be seen that the second order derivative of logarithmic negativity with respect to $p$ is always negative for all values of $p$ such that $\frac{1}{1+2^{n-1}}<p<1$ and for any value of $n$. Therefore, the analytical result is backed by the numerical results.

\end{appendices}
%

%
\end{document}